\documentclass[english,nosubfigcap,nolineno]{socg-lipics-v2019}
\newcommand{\PAPER}[1]{#1}
\newcommand{\SOCG}[1]{}

\usepackage[utf8]{inputenc}

\newcommand{\eps}{\varepsilon}
\newcommand{\polylog}{\mathop{\rm polylog}}
\newcommand{\Patrascu}{P\u{a}tra\c{s}cu}

\newcommand{\R}{\mathbb{R}}
\newcommand{\Ex}{\mathbb{E}}
\newcommand{\IGNORE}[1]{}

\usepackage{microtype,hyperref}
  \hypersetup{%
      breaklinks,%
      ocgcolorlinks, colorlinks=true,%
      urlcolor=[rgb]{0.25,0.0,0.0},%
      linkcolor=[rgb]{0.5,0.0,0.0},%
      citecolor=[rgb]{0,0.2,0.445},%
      filecolor=[rgb]{0,0,0.4},
      anchorcolor=[rgb]={0.0,0.1,0.2}%
   }

\EventEditors{Sergio Cabello and Danny Z. Chen}
\EventNoEds{2}
\EventLongTitle{36th International Symposium on Computational Geometry (SoCG 2020)}
\EventShortTitle{SoCG 2020}
\EventAcronym{SoCG}
\EventYear{2020}
\EventDate{June 23--26, 2020}
\EventLocation{Z\"{u}rich, Switzerland}
\EventLogo{socg-logo}
\SeriesVolume{164}

\title{Further Results on Colored Range Searching}

\author{Timothy M. Chan}{Department of Computer Science, University of Illinois at Urbana-Champaign, USA}{tmc@illinois.edu}{https://orcid.org/0000-0002-8093-0675}{
Supported in part by NSF Grant CCF-1814026.}

\author{Qizheng He}{Department of Computer Science, University of Illinois at Urbana-Champaign, USA}{qizheng6@illinois.edu}{}{}

\author{Yakov Nekrich}{Department of Computer Science, Michigan Technological University, Houghton, USA}{yakov.nekrich@googlemail.com}{}{}

\authorrunning{T.\,M. Chan, Q. He, and Y. Nekrich}
\Copyright{Timothy M. Chan, Qizheng He, and Yakov Nekrich}

\ccsdesc[100]{Theory of computation~Computational geometry}
\ccsdesc[100]{Theory of computation~Data structures design and analysis}

\keywords{Range searching, geometric data structures, randomized
incremental construction, random sampling, word RAM}

\hideLIPIcs  

\begin{document}

\maketitle

\begin{abstract}
We present a number of new results about range searching for colored (or ``categorical'') data:

\begin{enumerate}
\item For a set of $n$ colored points in three dimensions, we describe
randomized data structures with $O(n\mathop{\rm polylog}n)$ space that can
report the distinct colors in any query orthogonal range (axis-aligned box) in $O(k\mathop{\rm polyloglog} n)$ expected time, where $k$ is the number of distinct colors in the range, assuming that coordinates are in $\{1,\ldots,n\}$.  Previous data structures require $O(\frac{\log n}{\log\log n} + k)$ query time.  Our result also implies improvements in higher constant dimensions.
\medskip
\item Our data structures can be adapted to halfspace ranges in three dimensions (or circular ranges in two dimensions), achieving $O(k\log n)$ expected query time.  Previous data structures require $O(k\log^2n)$ query time.
\medskip
\item For a set of $n$ colored points in two dimensions, we describe a data structure with $O(n\mathop{\rm polylog}n)$ space that can answer colored ``type-2'' range counting queries: report the number of occurrences of every distinct color in a query orthogonal range.  The query time is $O(\frac{\log n}{\log\log n} + k\log\log n)$, where $k$ is the number of distinct colors in the range.  Naively performing $k$ uncolored range counting queries would require $O(k\frac{\log n}{\log\log n})$ time.
\end{enumerate}

\noindent
Our data structures are designed using a variety of techniques, including colored variants of randomized incremental construction (which may be of independent interest), colored variants of shallow cuttings, and bit-packing tricks.
\end{abstract}

\section{Introduction}

{\em Colored range searching\/} (also known as
``categorical range searching'', or ``generalized range searching'') have been
extensively studied in computational geometry since the 1990s.
For example, see the papers
\cite{BozanisKMT95,ChanN20,El-ZeinMN17,GagieKNP13,GangulyMNST19,GrossiV14,
GuptaJS95,GuptaJS96,GuptaJS97,GuptaJS99,
JL93,Kap08,KapSoCG06,LarsenP12,LarsenvW13,M03,Muth02,
Nekrich14,NekrichV13,PatilTSNV14,ShiJ05} and the
survey by Gupta et al.~\cite{GuptaSURVEY}.
Given a set of $n$ colored data points
(where the color of a point represents
its ``category''), the objective is to
build data structures that can provide
statistics or some kind of summary about the colors of
the points inside a query range.  The most basic types of
queries include:

\begin{itemize}
\item \emph{colored range reporting}: report all the distinct colors in the query range.
\item \emph{colored ``type-1'' range counting}: find the number of distinct colors in the query range.
\item \emph{colored ``type-2'' range counting}: report the number 
of points of color $\chi$ in the query range, for \emph{every}
color $\chi$ in the range.
\end{itemize}

\noindent
In this paper, we focus on colored range reporting and type-2
colored range counting.  Note that the output size in both instances
is equal to the number $k$ of distinct colors in the range, and we aim for
query time bounds that depend linearly on $k$, of the form $O(f(n) + kg(n))$.
Naively using an uncolored range reporting data structure
and looping through all points in the range would be too costly, since
the number of points in the range can be significantly larger than $k$.

\subsection{Colored orthogonal range reporting}
The most basic version of the problem is perhaps \emph{colored orthogonal range reporting}:
report the $k$ distinct colors inside an orthogonal range (an axis-aligned box).  It is not difficult to obtain an $O(n\polylog n)$-space data structure with
$O(k\polylog n)$ query time~\cite{GuptaSURVEY} for any constant dimension~$d$: one approach is to
directly modify the $d$-dimensional range tree~\cite{BergBOOK,PreShaBOOK}, and another approach is to
reduce colored range reporting to uncolored range emptiness~\cite{GuptaJS99} (by building a one-dimensional range
tree over the colors and storing a range emptiness structure at each node).
Both approaches require
$O(k\polylog n)$ query time rather than 
$O(\polylog n + k)$ as in traditional
(uncolored) orthogonal range searching: the reason
is that in the first approach, each color may be discovered
polylogarithmically many times, whereas in the second approach, each  discovered color costs us $O(\log n)$ range emptiness queries, each of which requires polylogarithmic time.

Even the 2D case remains open, if one is interested in optimizing logarithmic factors.  For example, Larsen and van Walderveen~\cite{LarsenvW13}
and Nekrich~\cite{Nekrich14} independently presented data structures with $O(n\log n)$ space and $O(\log\log U + k)$ query time in the standard word-RAM model, assuming that
coordinates are integers bounded by $U$.  The query bound is optimal, but
the space bound is not.  Recently, Chan and Nekrich~\cite{ChanN20} have improved the space bound to $O(n\log^{3/4+\eps}n)$ for an arbitrarily small constant $\eps>0$, while keeping $O(\log\log U + k)$ query time.

In 3D, the best result to date is by Chan and Nekrich~\cite{ChanN20}, who
obtained a data structure with $O(n\log^{9/5+\eps}n)$ space
and $O(\frac{\log n}{\log\log n} + k)$ query time.
The first step is a data structure for the case of 3D dominance (i.e., 3-sided) ranges:
as they noted, this case can be solved in $O(n)$ space
and $O(\frac{\log n}{\log\log n}+k)$ time by a known reduction \cite[Section 3.1]{Rahul17} to
3D 5-sided box stabbing~\cite{ChaICALP18}.
For 3D 5-sided box stabbing (or more simply, 2D 4-sided rectangle stabbing),
a matching lower bound of $\Omega(\frac{\log n}{\log\log n} + k)$ is known for
$O(n\polylog n)$-space structures, due to
\Patrascu~\cite{Patrascu11}.  A natural question then arises: is
$O(\frac{\log n}{\log\log n} + k)$ query time also tight for 3D colored dominance range reporting?

We show that the answer is no---the $O(\frac{\log n}{\log\log n})$ term can in fact be improved when $k$ is small.
Specifically, we present a randomized data structure for 3D colored dominance range
reporting with $O(n\log n)$ space and $O(\log\log U + k\log\log n)$ expected time in the standard word-RAM model.  (We use only Las Vegas randomization, i.e., the query algorithm is always correct; an oblivious adversary is assumed, i.e., the query range should be independent of the random choices made by the preprocessing algorithm.)  Combining with Chan and Nekrich's method~\cite{ChanN20}, we can then
obtain a data structure for 3D colored orthogonal range reporting with
$O(n\log^{2+\eps}n)$ space and $O(\log\log U + k\log\log n)$
expected query time.

An improved solution in 3D automatically implies improvements in any constant dimension $d>3$, by using standard range trees~\cite{BergBOOK,PreShaBOOK} to reduce the dimension, at a cost of about one logarithmic factor (ignoring $\log\log$ factors) per dimension.  This way, we obtain a data structure in $d$ dimensions with
$O(n\log^{d-1+\eps}n)$ space and $O(k(\frac{\log n}{\log\log n})^{d-3}\log\log n)$ query time.\footnote{In all reported bounds,
we implicitly assume $k>0$.  The $k=0$ case can be handled by
answering one initial uncolored range emptiness query.}
(Note that $O((\frac{\log n}{\log\log n})^{d-3}\log\log n)$ is the current best
query time bound for standard (uncolored) range emptiness~\cite{ChanLP11} for $O(n\polylog n)$-space structures on the word RAM\@.)

\subsection{Colored 3D halfspace range reporting}

An equally fundamental problem is \emph{colored halfspace range reporting}.
In 2D, an $O(n)$-space data structure with $O(\log n+k)$ query time
is known~\cite{AgarwalCTY09,GuptaSURVEY}.
In 3D, the current best result is obtained by applying a general reduction of colored
range reporting to uncolored range emptiness~\cite{GuptaJS99}, which yields
$O(n\log n)$ space and $O(k\log^2 n)$ query time~\cite{GuptaSURVEY}.
(An alternative solution with $O(n)$ space and $O(n^{2/3+\eps} + k)$ time
is also known, by reduction to simplex range searching.)
The 3D case is especially important, as 2D colored circular range
reporting reduces to 3D colored halfspace range reporting by the
standard lifting transformation.

We describe a randomized data structure with $O(n\log n)$ space
and $O(k\log n)$ expected query time for 3D halfspace ranges (and thus 2D circular ranges).  This is
a logarithmic-factor improvement over the previous query time bound.

\subsection{Colored 2D orthogonal type-2 range counting}
Finally, we consider colored orthogonal ``type-2'' range counting:
compute the number of occurrences of every color in a given orthogonal
range.  Despite the nondescript name, colored type-2 counting is quite
natural, providing more information than colored reporting, as we are generating an entire histogram.
The problem was introduced by Gupta et al.~\cite{GuptaJS95} (and more recently revisited by Ganguly et al.~\cite{GangulyMNST19} in external memory).  An old
paper by Bozanis et al.~\cite{BozanisKMT95} gave a solution in the 1D case
with $O(n)$ space and $O(\log n + k)$ query time, which implies 
a solution in 2D with $O(n\log n)$ space and $O(\log^2 n + k\log n)$ query
time.  Alternatively, to answer a colored type-2 counting query,
we can first answer
a colored range reporting query, followed by $k$
standard (uncolored) range counting queries, if we store each color class in
a standard range counting data structure; by known results on colored range reporting~\cite{ChanN20} and standard range counting~\cite{JaJaMS04}, this then yields
$O(n\log^{3/4+\eps}n)$ space and $O(k\frac{\log n}{\log\log n})$ query time.
Thus, in some sense, a type-2 counting query corresponds to 
``simultaneous'' range counting queries on multiple point sets.\footnote{
See \cite{AfshaniSTW14} for a different notion of ``concurrent'' range reporting.
}

We present a data structure for the problem in 2D with
$O(n\log^{1+\eps} n)$ space and $O(\frac{\log n}{\log\log n} + k\log\log n)$
query time in the standard word-RAM model.  As 2D standard (uncolored) range counting has an
$\Omega(\frac{\log n}{\log\log n})$ time lower bound for
$O(n\polylog n)$-space structures~\cite{Patrascu07}, our result shows, surprisingly, that
answering multiple range counting queries ``simultaneously'' are cheaper than 
answering one by one---we only have to pay $O(\log\log n)$ cost
per color!

\subsection{Techniques}

Our solutions for colored 3D dominance range reporting and 3D halfspace range reporting are based on similar ideas.  We in fact propose two different
methods.  

In the first method (Section~\ref{sec:first}), we solve the $k=1$ case (testing
whether a range contains only one color) by introducing a colored variant of 
\emph{randomized
incremental construction}; we then extend the solution to the general case
by a randomized one-dimensional range tree over the colors.  Along the way,
we prove a combinatorial lemma which may be of independent interest:
in a colored point set in 3D, if we randomly permute the color classes
and randomly permute the points within each color classes, and if we
insert the points in the resulting order, then
the convex hull undergoes $O(n\log n)$ structural changes in expectation.
(It is a well known fact, by Clarkson and Shor~\cite{clarkson1989applications}, that in the uncolored setting, if we insert points in a random order,
the 3D convex hull undergoes $O(n)$ structural changes in expectation.)

In the second method (Section~\ref{sec:second}), which is slightly more efficient, we solve the $k=1$ case differently, by
adapting known techniques for uncolored 3D halfspace range reporting~\cite{Chan00}
based on \emph{random sampling} (namely, conflict lists of 
lower envelopes of random subsets).
The approach guarantees only $\Omega(1)$
success probability per query (in the uncolored setting, \emph{shallow cuttings} can fix the problem, but they do not seem easily generalizable to the 3D colored setting).  Fortunately, we show that a solution
for the $k=1$ case with constant
success probability is sufficient
to complete the solution for the general case.

Our method for colored 2D type-2 orthogonal range counting (Section~\ref{sec:type2}) is technically the most involved.  It is
obtained by a nontrivial combination of several techniques, including the
\emph{recursive grid} approach of Alstrup, Brodal, and Rauhe~\cite{AlstrupBR00},
bit packing tricks, and 2D shallow cuttings.  Our work demonstrates yet again
the power of the recursive grid approach (see \cite{ChanLP11,ChanN20,ChaICALP18} for other recent examples).

\section{Colored 3D Halfspace Range Reporting: First Method}\label{sec:first}

In this and the next section, we describe our two methods for 3D halfspace ranges.  The case of 3D dominance ranges is similar and will be addressed later in Section~\ref{sec:orth}.  

\subsection{Combinatorial lemmas on colored randomized incremental construction}

Our first method relies on a simple combinatorial lemma related to a colored version of randomized incremental construction of 3D convex hulls (the uncolored version of the lemma, where all points are assigned different colors, is well known in computational geometry, from the seminal work by Clarkson and Shor~\cite{clarkson1989applications}):

\newcommand{\CH}{\textrm{CH}}

\begin{lemma}\label{lemma:CH}
Given a set $S$ of $n$ colored points in $\mathbb{R}^3$, if we first randomly permute the color classes, then for each color according to this order we simultaneously insert all points with that color, then the expected total number of structural changes 
to the convex hull is $O(n)$.
\end{lemma}

\begin{proof}
Consider a random permutation of the colors.  Let $C_i$ be the $i$-th color class, i.e., the set of all points with the $i$-th color in the permutation.  Let 
$m$ be the number of color classes. Let $V_i=\bigcup_{j=1}^i C_j$ contain all points with the first $i$ colors.
Let $\CH(V_i)$ denote the convex hull of $V_i$.
Let $\Delta_i^+$ be the set of all facets in $\CH(V_i)$ that are not in $\CH(V_{i-1})$, i.e.,  all hull facets created when we insert the $i$-th color class $C_i$.

For each $i$, we have $\mathbb{E}[|C_i|]=\frac{n}{m}$ and 
$\Ex[|V_i|]=\frac{in}{m}$. We use backwards analysis~\cite{Seidel93}.
Observe that $|\Delta_i^+|$ is bounded by the total degree of all points of $C_i$ in $\CH(V_i)$. The total degree over all points in $\CH(V_i)$ is
$O(|V_i|)$.
Conditioned on a fixed $V_i$, we have $\mathbb{E}[|\Delta_i^+|]=O(\frac{|V_i|}{i})$. 
So, unconditionally, $\Ex[|\Delta_i^+|]=O(\frac{n}{m})$.
Therefore, the expected total number of hull facets created is $\mathbb{E}[\sum_{i=1}^m |\Delta_i^+|]=O(n)$.
\end{proof}

The following refinement of the lemma further bounds the total amount of changes to the convex hull when we additionally insert points one by one in a random order within each color class. 
The proof is slightly trickier.  (The first lemma is already sufficient to bound the space of our new data structure, but the refined lemma will be useful in bounding the preprocessing time.)

\begin{lemma}\label{lemma:CH:refined}
Given a set $S$ of $n$ colored points in $\mathbb{R}^3$, if we first randomly permute the color classes, then randomly permute the points in each color class, and insert the points one by one according to this order, then the expected total number of structural changes of the convex hull is $O(n\log n)$.
\end{lemma}

\begin{proof}
Continuing the earlier proof,
let $\Delta_i^-$ be the set of all facets in $\CH(V_{i-1})$ that are not in $\CH(V_{i})$, i.e.,  all hull facets destroyed when we insert the $i$-th color class $C_i$.  Since the total number of facets destroyed is at most the total number of facets created, $\Ex[\sum_{i=1}^m |\Delta_i^-|]\le \Ex[\sum_{i=1}^m |\Delta_i^+|] = O(n)$.

Now, consider a random permutation of the points in $C_i$.
Let $V_{i,j}$ contain all points in $V_{i-1}$ and also the first $j$ points of $C_i$. 
Let $G_{i,j}$ be the subgraph formed by all edges of $\CH(V_{i,j})$ that are incident to the vertices of $C_i$.  Then every vertex $v$ in $G_{i,j}$ is either in $C_i$ or is incident to a facet of $\Delta_i^-$ (because if $v\not\in C_i$, then $v$ must be a vertex of $\CH(V_{i-1})$, and at least one of its incident facets in $\CH(V_{i-1})$ will be destroyed when $C_i$ is inserted).  Thus, $G_{i,j}$ has $O(|C_i| + |\Delta_i^-|)$ vertices, and since $G_{i,j}$ is a planar graph, it has $O(|C_i| + |\Delta_i^-|)$
edges.

Let $\Delta_{i,j}^+$ be the set of all facets in $\CH(V_{i,j})$ that are not in $\CH(V_{i,j-1})$, i.e., all hull facets created when we insert the $j$-th point in $C_i$.
We use backwards analysis again.
Observe that $|\Delta_{i,j}^+|$ is bounded by the degree of the $j$-th point in $C_i$ in $\CH(V_{i,j})$.
The total degree over all points of $C_i$ in $\CH(V_{i,j})$ is at most twice the number of edges in $G_{i,j}$.
Conditioned on a fixed $C_i$ and a fixed $V_{i,j}$, we thus have $\mathbb{E}[|\Delta_{i,j}^+|]=O\left(\frac{|C_i| + |\Delta_i^-|}{j}\right)$.
As the right-hand side does not depend on the local permutation of the color class $C_i$, the expectation holds conditioned only on the global permutation of the colors.  Unconditionally, the expected total number of hull facets created is
\[O\left(\Ex\left[\sum_{i=1}^m\sum_{j=1}^{|C_i|}\frac{|C_i| + |\Delta_i^-|}{j}\right]\right) = O\left(\Ex\left[\sum_{i=1}^m (|C_i| + |\Delta_i^-|)\log n\right]\right)
= O(n\log n).\]
\end{proof}

\subparagraph*{Remarks.}\ 

\begin{enumerate}
    \item The $O(n\log n)$ bound in the refined lemma is tight:
Consider $\frac{n}{2}$ points lying on the $xy$-plane in convex position, each assigned a different color.  In addition, add
$\frac{n}{2}$ points on the $z$-axis above the $xy$-plane, all with a common color $\chi_0$.
When we insert the color class for $\chi_0$, there are already $\Omega(n)$ points on the $xy$-plane with probability $\Omega(1)$. In an iteration where the next point we insert with color $\chi_0$ has larger $z$-coordinate than all previous points, the insertion would create $\Omega(n)$ new hull edges in expectation. By a well known analysis,
the expected number of such iterations is given by the Harmonic number, which is $\Theta(\log n)$.  This shows an $\Omega(n\log n)$ lower bound.
 \medskip
 \item The same argument holds for other geometric structures besides 3D convex hulls, e.g., Voronoi diagrams of 2D points and trapezoidal decompositions of 2D disjoint line segments.
 
\medskip
\item
 We can generalize the refined lemma to the setting when we have a hierarchy of color classes with $\ell$ levels, and we randomly permute the child subclasses of each color class.  (The refined lemma corresponds to the $\ell=2$ case.)  The bound becomes $O(n\log^{\ell-1}n)$.
 This result seems potentially relevant to 
 implementing randomized incremental constructions in a hierarchical external-memory model.



\end{enumerate}




\subsection{The $k=1$ case}

We now reveal how colored randomized incremental construction can help solve the colored range reporting problem.  We start with
the case $k=1$, i.e., we want to test whether there is only one color in the query range.  By an uncolored range search, we can find one point in the range (in $O(\log n)$ time for 3D halfspace ranges) and identify its color $\chi$.  Thus, the problem is to verify that all points in the range have the same color $\chi$.

Fix a total ordering of the colors.  It is easy to see that the problem reduces to two subproblems: for a given query color $\chi$, (i)~decide whether there exists a point in the range with color $<\chi$, and (ii)~decide whether there exists a point in the range with color $>\chi$.  By symmetry, it suffices to solve subproblem~(i). To this end, we imagine inserting the points in increasing order
of color, and maintaining a data structure for (uncolored) range emptiness for the points.
We can make this semi-dynamic data structure (which supports insertions only) \emph{persistent}.  Then we can solve subproblem (i) by querying a past version of the range emptiness data structure, right after all points with color $<\chi$ were inserted.

In the case of 3D upper halfspaces (lower halfspaces can be handled symmetrically), a range emptiness query reduces to finding an extreme point on the upper hull along a query direction, or equivalently, intersecting the lower envelope of the dual planes at a query vertical line.  By projection, this reduces to a planar point location query, answerable in $O(\log n)$ time by a linear-space data structure~\cite{BergBOOK,PreShaBOOK}.
However, we need a data structure that supports insertions, and in general this increases the query time (by an extra logarithmic factor via the standard ``logarithmic method''~\cite{BentleyS80}).

The key is to observe that the above approach works regardless of which total ordering of the colors we use.  Our idea is simply to use a \emph{random ordering} of the colors!  (For (ii), note that the reverse of a uniformly random ordering is still uniformly random.)  By Lemma~\ref{lemma:CH}, the upper hull undergoes $O(n)$ expected number of structural changes.  So is the dual lower envelope.  We can then apply a known dynamic planar point location method; for example, the method by
Chan and Nekrich~\cite{chan2018towards} achieves $O(\log n(\log\log n)^2)$ query time and $O(\log n\log\log n)$ amortized update time per change to the envelope.
The data structure can be made persistent, for example,
by applying Dietz's technique~\cite{Dietz89a}, with a $\log\log n$ factor penalty (the space usage is related to the total update time).  The final data structure supports queries in $O(\log n (\log\log n)^3)$ (worst-case) time and uses $O(n\log n(\log\log n)^2)$ expected space.  (Note that the space bound can be made worst-case, by repeating $O(1)$ expected number of times until a ``good'' ordering is found.)

\subparagraph*{Remark on preprocessing time.}
It isn't obvious how to efficiently insert an entire color class to the 3D convex hull, even knowing that the total number of structural changes is small.  To get good preprocessing time, we propose inserting points one by one within each color class, since Lemma~\ref{lemma:CH:refined} ensures that the number of changes to the convex hull is still near linear ($O(n\log n)$).
Several implementation options can then yield $O(n\polylog n)$ expected preprocessing time: (i)~we can use a general-purpose dynamic convex hull data structure~\cite{Chan10} (in the insertion-only case, the cost per update is $O(f\log^2n)$ where $f$ is the amount of structural changes); (ii)~we can adapt standard randomized incremental algorithms, e.g., handling the point location steps by using history DAGs~\cite{mulmuley1994computational} (this requires further randomized analysis); or (iii)~we can adapt standard randomized incremental algorithms, but handling the point location steps by using a known dynamic planar point location method~\cite{chan2018towards}.

\begin{theorem}\label{thm:one}
For $n$ colored points in $\R^3$, there is a data structure with
$O(n\polylog n)$ expected preprocessing time and $O(n\log n (\log\log n)^2)$ space that can test whether the number of colors in a query halfspace is exactly~$1$ in
$O(\log n (\log\log n)^3)$ time.
\end{theorem}



\subsection{The general case}

Previous papers~\cite{GuptaSURVEY,GuptaJS99}  (see also~\cite{ChazelleCPY86} in the uncolored case) have noted a straightforward black-box reduction of colored range reporting to the $k=0$ case (range emptiness),
essentially by using a one-dimensional range tree over the colors:
More precisely, we split the color classes into two halves.  We build
a data structure for $k=0$, and recursively build a data structure for the two halves.  Space usage increases by a logarithmic factor.  If the $k=0$ structure has $Q_0(n)$ query time,
the overall query time is $O(kQ_0(n)\log n)$, since at each of the $O(\log n)$ levels of recursion tree, $O(k)$ nodes are examined.

We present a new black-box reduction of colored range reporting to the $k\le 1$ case, which saves a logarithmic factor, by using a similar idea but with randomization. 

\begin{theorem}\label{thm:reduction}
Suppose that for $n$ colored points, there is a data structure with $P(n)$ (expected) preprocessing time and $S(n)$ space that can decide whether the number of colors in a query range is exactly $1$ in $Q_1(n)$ time. 
In addition, the data structure can decide whether the range is empty, and if not, report one point, in $Q_0(n)$ time. 
Then there is a randomized Las Vegas data structure with $O(P(n)\log n)$ expected preprocessing time and $O(S(n)\log n)$ space that can report all $k$ distinct colors in a query range in $O(k(Q_0(n)+Q_1(n)))$ expected time, assuming that $P(n)/n$ and $S(n)/n$ are nondecreasing.
\end{theorem}
\begin{proof}
We split the color classes into two parts, where
each color is \emph{randomly} assigned to one of the two parts.
We build the given $k=1$ structure and range emptiness structure, and recursively build a data structure for the two parts.  Space usage increases by a logarithmic factor (with high probability).

To answer a query, we test whether the range is empty or whether $k=1$.  If so, we are done.  Otherwise, we recursively query both parts.

Consider a query range that is independent of the random choices made by the data structure.
At the $i$-th level of the recursion tree, how many nodes are examined (in expectation)?  This question is analogous to the following: place $k$ balls randomly (independently) into $2^i$ bins; how many bins contain two or more balls?  The number is upper-bounded by the number of pairs of balls that are in the same bin.  Since the probability that a fixed pair of balls are placed in the same bin is $1/2^i$, the expected number of pairs is at most 
$k^2/2^i$. 

Thus, the expected number of nodes examined at the $i$-th level is at most $\min\{2^i, k^2/2^i\}$. 
The overall expected number of nodes examined is
\[O\left(\sum_i \min\{2^i, k^2/2^i\}\right)
\:=\: O\left(\sum_{i:\, 2^i\le k} 2^i + \sum_{i:\, 2^i>k} k^2/2^i \right) \:=\: O(k).\]
\end{proof}

Combining Theorems~\ref{thm:one} and~\ref{thm:reduction} yields:

\begin{theorem}\label{thm:report}
For $n$ colored points in $\R^3$, there is a randomized Las Vegas data structure with
$O(n\polylog n)$ expected preprocessing time and $O(n\log^2n (\log\log n)^2)$ space that can report all $k$ distinct colors in a query halfspace in
$O(k\log n (\log\log n)^3)$ expected time.
\end{theorem}

\section{Colored 3D Halfspace Range Reporting: Second Method}\label{sec:second}

We next describe a slightly better (and simpler) method for colored 3D halfspace range reporting.  

\subsection{The $k=1$ case}

The idea is to relax the $k=1$ subproblem and allow the query algorithm to occasionally be wrong (since we will be using randomization anyways for the general case).  The algorithm has constant error probability and can only make one-sided errors: if it returns ``yes'', we must have $k=1$.
We work in dual space: given a set of colored planes in $\R^3$, we want to decide whether the number of colors among the planes below a query point is exactly~1.

\newcommand{\LE}{\textrm{LE}}
\newcommand{\VD}{\textrm{VD}}
\newcommand{\D}{\Delta}

\subparagraph*{Preprocessing.}
Take a random sample $R$ of the planes, where each color class is included independently with probability $\frac{1}{2}$.
Take the lower envelope $\LE(R)$ of $R$, and consider
the vertical decomposition $\VD(R)$ of the region underneath $\LE(R)$.  (The vertical decomposition is defined as follows: we triangulate each face of $\LE(R)$ by joining each vertex to the bottom vertex of the face; for each triangle, we form the unbounded prism containing all points underneath the triangle.)
For each cell $\D\in\VD(R)$, let $L_\Delta$ denote the set of distinct colors among all planes intersecting $\Delta$ (the ``color conflict list'' of $\Delta$).  We store the list $L_\Delta$ if $|L_\Delta|\le c$ for a sufficiently large constant~$c$; otherwise, we mark $\D$ as ``bad''.

Clearly, the space usage is $O(n)$, since there are $O(n)$ cells in $\VD(R)$ and each list stored has constant size.  To bound the preprocessing time, we can generate (up to $c$ elements of) each list $L_\Delta$ by answering colored range reporting queries at the three vertices of $\Delta$, since a plane intersects $\Delta$ iff it is below at least one of the vertices of $\Delta$.  By previous results, these $O(n)$ colored range reporting queries take $O(n\polylog n)$ time.

In addition, for each color class, we store an (uncolored) range emptiness structure (i.e., a planar point location structure for the $xy$-projection of the lower envelope of the color class).  This takes $O(n)$ space in total.

\subparagraph*{Querying.}
Given a query point $q$, we find the cell $\Delta(q)$ of $\VD(R)$ containing $q$ in $O(\log n)$ time by planar point location (on the $xy$-projection of $\VD(R)$).  If the cell does not exist (i.e., $q$ lies above $\LE(R)$), or if the cell is bad, we return ``no''.
Otherwise, for each of the at most $c$ colors in the conflict list $L_{\D(q)}$, we test whether any plane below $q$ has that color by
querying the corresponding range emptiness structure in $O(\log n)$ time.  We return ``yes'' iff 
exactly one color passes the test.
The overall query time is $O(\log n)$.

The algorithm is clearly correct if it returns ``yes''.
Consider a fixed query point $q$, such that there is just one color $\chi$ among all planes below $q$.
The algorithm would erroneously return ``no'' in two scenarios: (i) when $q$ lies above $\LE(R)$,
or (ii) when $|L_{\D(q)}|>c$.  The probability of (i) is the probability that the color $\chi$ is chosen in the random sample $R$, which is $\frac{1}{2}$.  By the following lemma, and Markov's inequality, the probability of (ii) is at most $0.1$ (say) for a sufficiently large constant $c$.  This lemma directly follows from
Clarkson and Shor's technique~\cite{clarkson1989applications}
\SOCG{(see the full paper for a quick proof).}%
\PAPER{; for the sake of self-containment, we include a proof below.}

\begin{lemma}\label{lem:sample}
    For a fixed point $q$, we have $\Ex[|L_{\D(q)}|]=O(1)$.
\end{lemma}
\PAPER{
\begin{proof}
Let $D_\Delta$ denote the set of colors of the planes defining a cell $\Delta$.  Note that $|D_\Delta|=O(1)$.
Let $R'$ be another sample that includes each color class with probability $\frac{1}{4}$.  Then

\begin{eqnarray*}
1&\ge& \Ex[\text{\# of cells of $\VD(R')$ containing $q$}]\\
&=&\sum_{\Delta\ni q}\Pr[\Delta\text{ appears in }\VD(R')]\\
&=&\sum_{\Delta\ni q}\left(\tfrac{1}{4}\right)^{|D_\Delta|}\left(\tfrac{3}{4}\right)^{|L_\Delta|}
\ =\ \Theta\left(\sum_{\Delta\ni q}\left(\tfrac{3}{4}\right)^{|L_\Delta|}\right).
\end{eqnarray*}
On the other hand,
\begin{eqnarray*}
\mathbb{E}[|L_{\Delta(q)}|]&=&\sum_{\Delta\ni q}|L_\Delta|\cdot \mathrm{Pr}[\Delta\text{ appears in }\VD(R)]\\
&=&\sum_{\Delta\ni q}|L_\Delta|\left(\tfrac{1}{2}\right)^{|D_\Delta|}\left(\tfrac{1}{2}\right)^{|L_\Delta|}
\ =\ \Theta\left(\sum_{\Delta\ni q} |L_\Delta|\cdot \left(\tfrac{1}{2}\right)^{|L_\Delta|}\right).
\end{eqnarray*}
Therefore, $\mathbb{E}[|L_{\Delta(q)}|]=O(1)$.
\end{proof}
}

We conclude:
\begin{theorem}\label{thm:one2}
For $n$ colored points in $\R^3$, there is a randomized Monte Carlo data structure with $O(n\polylog n)$ preprocessing time and $O(n)$ space that decides whether the number of colors in a query halfspace is exactly $1$ in $O(\log n)$ time; if the actual answer is true, the algorithm returns ``yes'' with probability $\Omega(1)$, else it always returns ``no''.
\end{theorem}

\subparagraph*{Remarks.}
The method can be viewed as a variant of Chan's random-sampling-based method for
uncolored 3D halfspace range reporting~\cite{Chan00}.
In the uncolored setting, errors can be completely avoided by replacing lower envelopes of samples with
\emph{shallow cuttings}~\cite{Matousek92}, but it is unclear how to do so in the colored setting.

\subsection{The general case}

Finally, to solve the general problem, we use a variant of Theorem~\ref{thm:reduction} that tolerates one-sided errors in the given $k=1$ data structure.

\begin{theorem}\label{thm:reduction2}
Suppose that for $n$ colored points, there is a randomized Monte Carlo data structure with $P(n)$ (expected) preprocessing time and $S(n)$ space that decides whether the number of colors in a query range is exactly $1$ in $Q_1(n)$ time;
if the actual answer is true, the algorithm returns ``yes'' with probability $\Omega(1)$, else it always returns ``no''.
In addition, the data structure can decide whether the range is empty, and if not, report one point, in $Q_0(n)$ time (without errors). 
Then there is a randomized Las Vegas data structure with $O(P(n)\log n)$ expected preprocessing time and $O(S(n)\log n)$ space that can report all $k$ distinct colors in a query range in $O(k(Q_0(n)+Q_1(n)))$ expected time, assuming that $P(n)/n$ and $S(n)/n$ are nondecreasing.
\end{theorem}
\begin{proof}
We use the same approach as in the proof of Theorem~\ref{thm:reduction}.  In the query algorithm, if the range is empty or the $k=1$ structure returns ``yes'', we are done; otherwise, we recursively query both parts.

To analyze the query time, we say that a node in the recursion tree is \emph{bad} if the number of colors in the query range at the node is exactly~1.  Our earlier analysis shows that the expected total number of non-bad nodes visited is $O(k)$.  However, because of the possibility of one-sided errors, the query algorithm may examine some bad nodes.
For each bad node $v$ visited by the query algorithm, we charge $v$ to its lowest ancestor $u$ that is not bad.  Then for a fixed node $u$, we may have up to two paths of nodes charged to $u$.
The expected number of nodes charged to a fixed node $u$ is at most $O(\sum_i (1-\Omega(1))^i)=O(1)$.  We conclude that the expected total number of nodes visited is $O(k)$.
\end{proof}

Combining Theorems~\ref{thm:one2} and~\ref{thm:reduction2} yields:

\begin{theorem}\label{thm:report2}
For $n$ colored points in $\R^3$, there is a randomized Las Vegas data structure with
$O(n\polylog n)$ expected preprocessing time and $O(n\log n)$ space that can report all $k$ distinct colors in a query halfspace in
$O(k\log n)$ expected time.
\end{theorem}

\section{Colored 3D \SOCG{Dominance} \PAPER{Orthogonal} Range Reporting}\label{sec:orth}

Both methods can be adapted to solve the colored 3D dominance range reporting problem: here, we want to report the distinct colors of all points inside a 3-sided range of the form $(-\infty,q_1]\times (-\infty,q_2]\times (-\infty,q_3]$.  Equivalently, we can map input points $(p_1,p_2,p_3)$ to orthants $[p_1,\infty)\times [p_2,\infty)\times [p_3,\infty)$, and the problem becomes reporting the distinct colors among all orthants containing a query point $q=(q_1,q_2,q_3)$.
By replacing values with their ranks, we may assume that all coordinates are in $\{1,\ldots,n\}$ (in a query, an initial predecessor search to reduce to rank space requires
an additional $O(\log\log U)$ cost by van Emde Boas trees).  We assume the standard word-RAM model.

In the first method, the combinatorial lemmas on colored randomized incremental constructions can be extended to the union of the orthants (a ``staircase polyhedron'').  In fact, by a known transformation involving an exponentially spaced grid~\cite{ChanLP11,PachT11}, orthants can be mapped to halfspaces and a union of orthants can be mapped to a halfspace intersection, or in the dual, a 3D convex hull.
For the $k=1$ structure, we not only randomly permute the color classes but also randomly permute the points inside each color class, and maintain the union of the orthants as points are inserted one by one.  Instead of using persistence, we reduce to static 3D point location: we insert in reverse order, and as a new orthant is inserted, we create a region for the newly added portion of the union (i.e., the new orthant minus the old union).
Identifying the smallest color of the orthants containing $q$ (to solve subproblem (i)) reduces to locating the region containing $q$.  The expected total size of these regions is $O(n\log n)$ by Lemma~\ref{lemma:CH:refined}; we can further subdivide each of these regions into boxes (by taking a vertical decomposition), without asymptotically increasing the total size.
Known results on orthogonal point location in a 3D subdivision of (space-filling) boxes~\cite{deberg1995two,ChaICALP18} then give $O((\log\log n)^2)$ query time and space linear in the size of the subdivision.
Thus, the final data structure for the general case has
$O(n\log^2n)$ space and $O(\log\log U + k(\log\log n)^2)$ expected query time.

In the second method, we replace lower envelopes with unions of orthants.  The only main change is that planar point location queries for orthogonal subdivisions now cost $O(\log\log n)$ time by Chan's result~\cite{chan2013persistent} instead of $O(\log n)$.
Thus, the final data structure has
$O(n\log n)$ space and $O(\log\log U + k\log\log n)$ expected time.

\begin{theorem}\label{thm:dom}
For $n$ colored points in $\R^3$, there is a randomized Las Vegas data structure with
$O(n\polylog n)$ expected preprocessing time and $O(n\log n)$ space that can report all $k$ distinct colors in a query dominance range in
$O(\log\log U + k\log\log n)$ expected time.
\end{theorem}

\SOCG{
We can extend the result of Theorem~\ref{thm:dom} to orthogonal ranges with more sides or to $d>3$ dimensions.  See the full paper for more details.
}

\PAPER{

From a colored dominance (i.e., 3-sided) range reporting structure, one can obtain colored reporting structures for orthogonal ranges with more sides by using a standard range-tree-based transformation which does not increase the query time but increases the space by a logarithmic factor per side added (i.e., a $\log^3 n$ factor for general 6-sided ranges).  Chan and Nekrich~\cite{ChanN20} recently gave a method that produces better space bounds.  Combining methods yields:

\begin{corollary}
For $n$ colored points in $\R^3$, there is a randomized Las Vegas data structure with
$O(n\log^{1+\eps} n)$ space that can report all $k$ distinct colors in a query 5-sided orthogonal range in
$O(\log\log U + k\log\log n)$ expected time.  For general 6-sided orthogonal ranges, the space bound increases to $O(n\log^{2+\eps}n)$.
\end{corollary}
\begin{proof}
Starting with a colored dominance structure with $O(n)$ space and $O(\frac{\log n}{\log\log n} +k)$ query time,
Chan and Nekrich~\cite[Theorem~3.6]{ChanN20} obtained a data structure for the 5-sided case with $O(n\log^\eps n)$ space and
$O(\frac{\log n}{\log\log n} + \frac{k\log^{O(1)}k}{\log n} + k)$ query time.  Starting with the new colored dominance structure with $O(n\log n)$ space and $O(\log\log U +k\log\log n)$ expected query time, it can be checked that the new space bound is $O(n\log^{1+\eps}n)$ and the new expected query time bound is $O(\log\log U + \frac{k\log^{5+\eps}k}{\log n} + k\log\log n)$.  If $k\le\log n$, the middle term disappears.

On the other hand, if $k>\log n$, we can switch to Chan and Nekrich's other data structure with $O(n\log^{4/5+\eps}n)$ space and $O(\frac{\log n}{\log\log n} + k)$ query time, which is $O(k)$.

As mentioned, we can transform a 5-sided structure into a 6-sided one by using range trees, with an extra logarithmic factor in space.
\end{proof}

In higher constant dimensions, one can use a $b$-ary range tree, which increases the space by a $b^{O(1)}\log n$ factor per dimension and query time by a $\log_b n$ factor per dimension.  Setting $b=\log^{O(\eps)}n$ (and replacing the initial predecessor search cost $O(\log\log U)$ with just $O(\log n)$) then yields:

\begin{corollary}
For $n$ colored points in $\R^d$ for a constant~$d>3$, there is a randomized Las Vegas data structure with
$O(n\log^{d-1 +\eps} n)$ space that can report all $k$ distinct colors in a query orthogonal range in
$O(k (\frac{\log n}{\log\log n})^{d-3}\log\log n)$ expected time.
\end{corollary}

}







\section{Colored 2D Orthogonal Type-2 Range Counting}\label{sec:type2}

\newcommand{\clist}{\mathrm{clist}}
\newcommand{\cE}{{\cal E}}
\newcommand{\NULL}{\textit{NULL}}
\newcommand{\Dom}{\textit{Dom}}

Our solution for orthogonal type-2 range counting is described in stages. First we consider the \emph{capped} variant of type-2 range counting. A capped query returns the correct answer if the number of colors $k$ in the query range does not exceed $\log^3 n$. If $k> \log^3 n$, the answer to the capped query is $\NULL$. Capped queries in the case when the query range is bounded on $2$ sides are considered in Section~\ref{sec:dominance}. 
\SOCG{%
We extend the solution to $3$-sided and $4$-sided queries,
as well as for the case when the number of colors can be arbitrarily large, in the full paper.
}%
\PAPER{%
We extend the solution to $3$-sided and $4$-sided queries in Sections~\ref{sec:3sided} and~\ref{sec:4sided} respectively. Finally we describe the solution for the case when the number of colors can be arbitrarily large in Section~\ref{sec:general}.
}

\subsection{Capped $2$-Sided Queries}
\label{sec:dominance}
 With foresight, we will solve the more general weighted version of this problem. Each point in $S$ is also assigned a positive  integer \emph{weight}.  For a $2$-sided query range $Q$, we want to identify all colors that occur in $Q$; for each color we report the total weight of all its occurrences in $Q$.

We will denote by $n$ the total weight of all points in $S$;
we will denote by $m$ the total number of points in $S$. 
\PAPER{The reason for this change of notation will be clear in Section~\ref{sec:3sided}.
}
We prove the following result:
\begin{lemma}
  \label{lemma:2dcapped-domin}
  Let $S$ be the set of $m$ points in $\R^2$ with total weight $n\ge m$. There exists a data structure that uses $O(m(\log \log n)^2)$ words of space and supports $2$-sided capped type-2 counting queries in $O(\log n/\log \log n + k\log\log n)$ time. 
\end{lemma}

Our data structure is based on the recursive grid approach~\cite{AlstrupBR00}. The set of points is recursively sub-divided into vertical slabs (or columns) and horizontal slabs (or rows). 

\subparagraph*{Data Structure.}  Let $\tau=\log^3 n_0$ where $n_0$ is the total weight of all points in the global data set (thus $\tau$ remains unchanged on all recursion levels). We divide the set of points into $\sqrt{n/\tau}$ columns so that either the total weight of all points in a column is bounded by $O(\sqrt{n\tau})$ or a column contains only one point.  This division can be obtained by scanning the set of points in the left-to-right order. We add points to a column $C_i$ for $i=1,2,\ldots$ by repeating the following steps: (1) if the weight of the next point $p$ exceeds $\sqrt{n\tau}$, we increment $i$, (2) we add $p$ to $C_i$, and (3) if the total weight of $C_i$ exceeds $\sqrt{n\tau}$, we increment $i$. 

Thus either the total weight of a column exceeds $\sqrt{n\tau}$ or the next column contains a point of weight at least $\sqrt{n\tau}$. Hence the number of columns is $O(\sqrt{n/\tau})$. 
We also divide the set of points into rows satisfying the same conditions.  Let $p_{ij}=(x_i,y_j)$ denote the point where the upper boundary of the $j$-th row intersects the right boundary of the $i$-th column. Let $\Dom(i,j)=[0,x_i]\times [0,y_j]$, denote  the range dominated by $p_{ij}$. If $\Dom(i,j)$ contains at most $\tau$ distinct colors, we store the list $L_{ij}$ of colors that occur in $\Dom(i,j)$. For every color in $L_{ij}$ we also keep the number of its occurrences in $\Dom(i,j)$. If the range $\Dom(i,j)$ contains more than $\tau$ different colors, we set $L_{ij}=\NULL$.  Thus $L_{ij}$ provides the answer to a capped type-2 counting query on $[0,x_i]\times [0,y_j]$.

Every row/column of weight at least $\tau^2$ that contains more than one point is recursively divided in the same way as explained above. If the total weight of all points is smaller than $\tau^2$, we can answer a type-2 range counting query in $O(k)$ time.
\PAPER{See Appendix~\ref{sec:small}.}

\subparagraph*{Slow Queries.} A query $[0,a]\times [0,b]$ is answered as follows. We identify the column $C_{i+1}$ containing $a$ and the row $R_{j+1}$ containing $b$. The query is then divided into the middle part 
$[0,x_i]\times [0,y_j]$, the upper part $[0,a]\times [y_j,b]$ and the right part $[x_i,a]\times [0,y_j]$. The answer to the middle query is stored in the pre-computed list $L_{ij}$. 
The upper query is contained in the row $R_{j+1}$ and the right query is contained in the column $C_{i+1}$. Hence we can answer the upper and the right query using data structures on $R_{j+1}$ and $C_{i+1}$ respectively. If $L_{ij}=\NULL$, we return $\NULL$ because the number of colors in the query range exceeds $\log^2 n$; if the answer to a query on $C_{i+1}$ or $R_{j+1}$ is $\NULL$, we also return $\NULL$. Otherwise, we merge the answers to the three queries. The resulting list $L$ can contain up to three items of the same color because the same color can occur in the left, right, and middle query. Since the items in $L$ are sorted by color, we can scan $L$ and compute the total number of occurrences for each color in time proportional to the length of~$L$.

The total query time is given by the formula
$Q(n,k)= O(k)+Q(\sqrt{n\tau},k_1) + Q(\sqrt{n\tau},k_2)$ where $k$ is the number of colors in the query range and $n$ is the total weight of all points. We denote by $k_1$ (resp.\ $k_2$) the total number of colors reported by the query on $R_{j+1}$ (resp. $C_{i+1}$). There are at most $2^i$ recursive calls at level $i$ of recursion. The total weight of points at recursion level $i$ is bounded by $n^{1/2^i}\log^{3(1-1/2^i)}n$.  Hence the number of recursion levels is bounded by $\ell=\log \log n - 2\log\log\log n$ and the total query time is $\sum_{i=1}^{\ell} 2^i \cdot k =O(k\cdot (\log n/\log \log n))$.

\subparagraph*{Fast Queries.}
We can significantly speed-up queries using the following approach. We keep colors of all points in a column/row in the rank space. Thus each point column or row on the $l$-th level of recursion contains $O(n^{1/2^l})$ points. Hence for any list $L_{ij}$ on the $l$-th recursion level we  can keep each color and the number of its occurrences in $\Dom(i,j)$ using $O((1/2^l)\log n)$ bits. 

As explained above, the query on recursion level $l$ is answered by merging 
three lists: the list $L_{ij}$ that contains the pre-computed answer to the middle query, the list of colors that  occur in the right query, and the list of colors that occur in the upper query. Every list occupies $O(k/2^l)$ words of $\log n$ bits. Hence we can merge these lists in $O(k/2^l)$ time using table 
look-ups.  Hence the total query time is 
$\sum_{i=1}^{\ell} 2^i \cdot \lceil k/2^i\rceil =O(\log n/\log\log n + k\cdot \log \log n)$.

\subparagraph*{Color Encoding.}
In order to merge lists efficiently we must be able to convert the color encoding for the slab $V^l$  into color encoding for the slab $V^{l-1}$ that contains $V^l$. Moreover the conversion should be performed in $O(k/2^l)$ time, i.e., in sub-constant time per color.  For this purpose we introduce the concept of \emph{colored $t$-shallow cutting} that adapts the concept of shallow cutting to the muti-color scenario. A colored $t$-shallow cutting for a set of points $S$ is the set of $O(|S|/t)$ cells.  Each cell is a rectangle with one corner in the point $(0,0)$.  Each cell contains points of  at most $2t$ different colors. If some point $q$ is not contained in any cell of the $t$-shallow cutting, then $q$ dominates points of at least $t$ different colors.

A colored $t$-shallow cutting can be constructed using the staircase approach, see e.g.,~\cite{VengroffV96}. We start in the point $(0,x_{\max}+1)$ where $x_{\max}$ is the largest $x$-coordinate of a point in $S$. We move $p$ in the $+y$ direction until $p$ dominates $2t$ different colors. Then we move $p$ in the $-x$ direction until $p$ dominates $t$ different colors. We alternatingly move $p$ in $+y$ and $-x$ directions until the $x$-coordinate of $p$ is $0$ or the $y$-coordinate of $p$ is $y_{\max}+1$ where $y_{\max}$ is the largest $y$-coordinate of any point in $S$. Each point where we stopped moving $p$ in $+y$ direction and started moving $p$ in the $-x$ direction is the upper right corner of some cell. We can show that the number of cell does not exceed $O(|S|/t)$:
Let $c_i=(x_i,y_i)$ and $c_{i+1}=(x_{i+1},y_{i+1})$ denote two consecutive corners (in the left-to-right order) of a $t$-shallow cutting. Consider all points $p=(x_p,y_p)$ such that
$x_i\le x_p\le x_{i+1}$ and $y_p\le y_{i+1}$. By our construction, points $p$ that satisfy these conditions have $t$ different colors. Hence there are at least $t$ such points $p$ and we can assign $t$ unique points to every corner of a colored $t$-shallow cutting. Hence the number of corners is $O(n/t)$. 
\begin{figure}
  \centering
  \includegraphics[width=.6\textwidth]{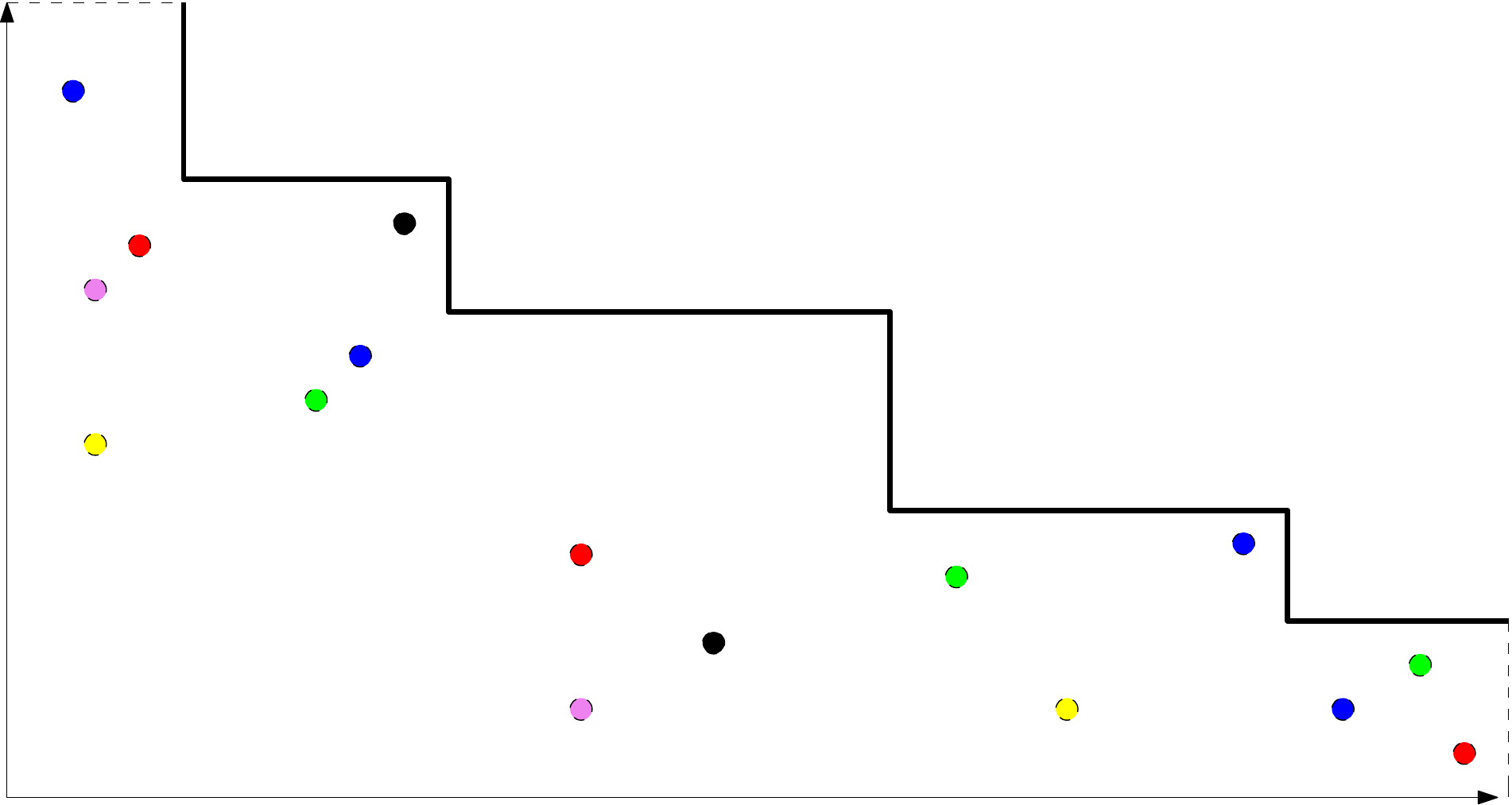}
  \caption{Example of colored $t$-shallow cutting for $t=3$.}
  \label{fig:colored-shcutting}
\end{figure}

For each slab  $V^l$ on any recursion level $l$, we construct colored $t$-shallow cuttings for $t=2$, $4$, $\ldots$, $\tau$. For every cell $c_j$ of each shallow cutting we create the list $\clist(c_j)$ of colors that occur in $c_j$. Colors in $\clist(c_j)$ are stored in increasing order. For each color we store its rank in $V^l$  and  its rank in the slab $V^{l-1}$ that contains $V^l$. Consider a $2$-sided query to a slab $V^l$ on recursion level $l$. The answer to this query is a sorted list $LIST(q)$ of $t$ colors in the rank space of $V^l$.  If $t< \log^2 n$, then the 2-sided query range is contained
in some  cell $c_j$ of  the colored $2^{\lceil\log t\rceil}$ shallow cutting. Using
$\clist(c_j)$ we can convert colors in $LIST(q)$ into the rank space of $V^{l-1}$ where $V^{l-1}$ is the slab that contains $V^l$. The conversion is based on a universal look-up table and takes $O(t/2^l)$ time.

\SOCG{
This result can be extended to $3$-sided and $4$-sided capped queries using divide-and-conquer on range trees and the recursive grid approach. The space usage  of the data structure for capped $4$-sided queries is $O(m\log^{\eps} n)$. Finally the range tree on colors enables us to get rid of the restriction on the number of colors, but the space usage of the data structure is increased by $O(\log n)$ factor. The complete description can be found in
the full paper.
}


\PAPER{



\subsection{Capped $3$-Sided Queries}
\label{sec:3sided}
$2$-sided queries can be extended to the case of $3$-sided queries using standard divide-and-conquer on a range tree. This technique was previously used in the context of color reporting in e.g.,  \cite{AGM02,Nekrich14}.
\begin{lemma}
  \label{lemma:2dcapped-3sid1}
  Let $S$ be the set of $m$ points with total weight $n\ge m$. There exists a data structure that uses $O(m\log m (\log \log n)^2)$ words of space and supports $3$-sided capped type-2 counting queries in $O(\log n/\log \log n + k\log\log n)$ time. 
\end{lemma}
We create the range tree on $x$-coordinates of points. All points in a node $u$ are stored in the data structure of Lemma~\ref{lemma:2dcapped-domin} that supports $2$-sided queries. Given  a query range $[a,b]\times [0,c]$, we find the lowest common ancestor $u$ of the leaves holding $a$ and $b$. Let $u_l$ and $u_r$ denote the left and the right children of $u$. We answer  type-2 range counting queries $[0,b]\times [0,c]$ on $S(u_r)$ and $[a,+\infty)\times [0,c]$ on $S(u_l)$. The answer to each query is a sorted list of colors occurring in $[0,b]\times [0,c]\cap S(u_r)$ and $[a,+\infty)\times [0,c]\cap S(u_r)$ respectively.
Since these lists contain colors in sorted order, they can be merged in $O(k)$ time.

\subparagraph*{Space-Efficient Data Structure.}
The space usage can be reduced to $O(m\log^{\eps}n)$ using the recursive \emph{lopsided} grid approach~\cite{ChanLP11}. 

Our data structure can be viewed as a tree with node degree $O(A)$, for a parameter $A=2^{\log^{1-\eps}n}$, on $x$-coordinates of points. The tree is similar to a range tree, but we take into consideration that points are weighted.  All points are stored in the root node.
If the total weight of the set $S(u)$ stored in a node $u$ exceeds $A$ and $S(u)$ contains more than one point, we distribute the points from $S(u)$ among the $O(A)$ children of $u$.
We guarantee that either the total weight of all points in a child node does not exceed $O(n_u/A)$, where $n_u$ is the total weight of points in a node $u$, 
or the child node contains only one point. Points from $S(u)$ can be distributed among the child nodes using the method described in Section~\ref{sec:dominance}. 

Each set $S(u)$ stored in an internal  node $u$ is divided into columns and rows. Columns correspond to children of $u$: a point is stored in the $i$-th column $C_i$ if it is stored in the $i$-th child of $u$.  $S(u)$ is divided into $O(n_u/(A\cdot \tau\cdot\log m))$ rows $R_i$, so that either the total weight of all points in a row is $O(\tau\cdot \log m\cdot A)$ or the row consists of only one point. Again we can use the method from Section~\ref{sec:dominance} to divide $S(u)$ into rows.

We keep the following additional data structures in tree nodes.  For every grid cell $G_{ij}=C_i\cap R_j$ we identify up to $\tau$  distinct colors that occur in $G_{ij}$. For every such color we keep a single point $p(G_{ij},\alpha)$ in $D_t$. The weight of $p(G_{ij},\alpha)$ is the total weight of all points with color $\alpha$ in $G_{ij}$. These points are stored in a data structure $D_t$ implemented as in Lemma~\ref{lemma:2dcapped-3sid1}.  If a grid cell contains more than $\tau$ colors, we say that this cell is \emph{marked}.  For every column we store a data structure supporting $2$-sided queries. We keep a recursively defined data structure for every row that contains more than  $2^{2\log^{\eps/2}n}$ points.  
We also keep a recursively defined data structure for each set $S(u)$, such that $u$ is a leaf node and  $S(u)$ contains more than $2^{2\log^{\eps/2}n}$ points. If the data structure contains at most $2^{2\log^{\eps/2}n}$ points, then we implement it as described in Lemma~\ref{lemma:2dcapped-3sid1}.

A query $Q=[a,b]\times [0,c]$ is answered as follows. We identify the lowest common ancestor $u$ of leaves that hold $a$ and $b$. Since $[a,b]\times [0,c]\cap S \subseteq S(u)$, it suffices to answer the query on $S(u)$. Let $C_l$ and $C_r$ denote columns that contain $a$ and $b$. Let $R_j$ denote the row that contains $c$. The query is divided into four parts.  We answer $2$-sided queries $Q\cap C_l$ and $Q\cap C_r$ on columns that contain $a$ and $b$. We answer a $3$-sided query $Q\cap R_j$ using the recursive data structure in a row $R_j$.
Let the remaining part of the query $Q'= Q\cap (S(u)\setminus (C_l\cup C_r\cup R_j))$ be called the middle query.  Sides of the middle query correspond to column and row boundaries. If the middle query contains a  marked cell, we report \NULL. Otherwise we answer the middle query using the data structure $D_t$.

Let $S(n,m)$ denote the space usage of the data structure in bits and let $R(n)=S(n,m)/m$. 
Then  $R(n)= \log^{1+\eps}n + (1+\log^{\eps}n) R(2^{1-\log^{\eps}n})$. The recursion depth is $O(1)$ and the space usage in the base case is $R(n)=O(\log^{2\eps}n(\log\log n)^2)$.   Hence $R(n)=O(\log^{1+2\eps} n(\log\log n)^2)$. By adjusting the constant $\eps$, the  total space usage of our data structure in words is $O(m\log^{\eps}n)$.

\begin{lemma}
  \label{lemma:2dcapped-3sid2}
  Let $S$ be the set of $m$ points in $\R^2$ with total weight $n\ge m$. There exists a data structure that uses $O(m\log^{\eps} n)$ words of space and supports $3$-sided capped type-2 counting queries in $O(\log n/\log \log n + k\log\log n)$ time. 
\end{lemma}

\subsection{Capped $4$-Sided Queries}
\label{sec:4sided}
The result of Lemma~\ref{lemma:2dcapped-3sid2} can be extended to $4$-sided queries using the same technique as in Lemma~\ref{lemma:2dcapped-3sid1}.
\begin{lemma}
  \label{lemma:2dcapped-4sid1}
  Let $S$ be the set of $m$ points in $\R^2$ with total weight $n\ge m$. There exists a data structure that uses $O(m \log m \log^{\eps} n)$ words of space and supports $4$-sided capped type-2 counting queries in $O(\log n/\log \log n + k\log\log n)$ time. 
\end{lemma}

\begin{lemma}
  \label{lemma:2dcapped-4sid2}
  Let $S$ be the set of $m$ points in $\R^2$ with total weight $n\ge m$. There exists a data structure that uses $O(m \log^{\eps} n)$ words of space and supports $4$-sided capped type-2 range counting queries in $O(\log n/\log \log n + k\log\log n)$ time. 
\end{lemma}

We use the same recursive grid approach as in Section~\ref{sec:dominance} and keep the data structure of Lemma~\ref{sec:3sided} for every row and every column. We also keep the data structure $D_t$ constructed as follows.  For every grid cell $G_{ij}=C_i\cap R_j$ we identify $\tau$  colors that occur in $G_{ij}$. For every such color we keep a single point $p(G_{ij},\alpha)$ in $D_t$. The weight of $p(G_{ij},\alpha)$ is the total weight of all points with color $\alpha$ in $G_{ij}$. If $G_{ij}$ contains more than $\tau$ colors, we say that this cell is marked.

The total weight of all points in $D_t$ is bounded by $n$ and the number of points is bounded by $m=n/\tau$. Hence, by Lemma~\ref{lemma:2dcapped-4sid1} the space used by $D_t$ is bounded by $O(m\log m \log^{\eps} n)=O(n\log^{\eps}n)$ words. Data structures for three-sided queries also use $O(n \log^{\eps} n)$ words. Summing over all recursion levels, 
the total space usage of all data structures is $\sum_{i=1}^{\ell} O(n\log^{1+\eps} n)=O(n\log^{1+\eps}n\log\log n)$ bits. 

If a query is entirely contained in one slab, we recursively answer the query using the data structure for that slab. If a query intersects more than one column or more than one row, we can represent the query as a union of at most four three-sided queries on slabs and at most one $4$-sided  query. This decomposition of a query into four parts is almost the same as the decomposition used in the proof of Lemma~\ref{lemma:2dcapped-domin}.  If the middle  query contains at least one marked cell, we  return $\NULL$.
Otherwise we can answer the middle query using data structure $D_t$. We can answer three-sided queries using the data structure of Lemma~\ref{lemma:2dcapped-3sid2}. The list of all colors in a $4$-sided range and their occurrences is computed by merging the answers to 3-sided queries and the query on $D_t$.

\subsection{General Case}
\label{sec:general}
Capped type-2 counting queries can be extended to the general case by constructing the range tree on colors. This technique was recently applied to color reporting queries~\cite{ChanN20} and was used earlier to answer  circular range reporting queries~\cite{ChazelleCPY86}.  We associate a set of colors with every node of the range tree. The set of all colors is associated to the root node; the set of colors in each node $\nu$ is divided into two equal parts that are associated to the left and the right children of $\nu$. Let $S(\nu)$ denote the set of all points $p$, such that the color of $p$ is associated to the node $\nu$.  

We keep the data structure of Lemma~\ref{lemma:2dcapped-4sid2} for each $S(\nu)$. Additionally we store a data structure $\cE$ that enables us to estimate the number of colors in $S(\nu)\cap Q$ for any $4$-sided range $Q$: if $Q\cap S(\nu)$ contains at most $\log^2 n$ colors, the data structure returns {\tt yes}. If $Q\cap S(\nu)$ contains more than $\log^3 n$ colors, the data structure returns {\tt no}. If the number of colors is between $\log^2 n/2$ and $\log^3 n$ the answer can be either {\tt yes} or {\tt no}. 

\begin{lemma}
  \label{lemma:struct-e}
Data structure $\cE$ can be implemented in $O(n\log^{\eps} n)$ space so that queries are supported in $O(\log n/\log\log n)$ time.
\end{lemma}
We construct a standard range tree with node degree $\log^{\eps/2}n$ on $x$-coordinates of points. For every node $u$ of the range tree and any $1 \le j\le \log^{\eps/2}n$  let $M(u,i,j)$ denote the set of points stored in children $u_i,\ldots,u_j$ of $u$. We keep $M(u,i,j)$ in the data structure from \cite{El-ZeinMN17} that uses $O(|M(u,i,j)|)$ bits and answers one-dimensional approximate range counting queries (with respect to $y$-coordinates) in $O(1)$ time. 

A query  range $Q=[a,b]\times [c,d]$ can be represented as a union of $O(\log n/\log\log n)$ one-dimensional queries: For any $[a,b]$ there are $O(\log n/\log\log n)$ sets $M(u_t,i_t,j_t)$ such that a point $p$ is in $[a,b]\times [c,d]$ if and only if  $p.y\in [c,d]$ and $p\in M(u_t,i_t,j_t)$ for some $t$. 
Suppose that the number of distinct colors of all points $p$ satisfying $p.y\in [c,d]$ and $p\in M(u_t,i_t,j_t)$ is between $\alpha_t$ and $2\alpha_t$. 
If $\alpha_t< \log^2 n$ for all $t$, we return {\tt yes}. In this case the total number of colors in $Q$ is smaller than $\log^3 n$. If $\alpha_t> \log^2 n$ for at least one $t$, we return {\tt no}. Obviously in this case there are at least $\log^2 n$ distinct colors in $Q$.

Now we return to the type-2 range counting problem. 
To answer a query $Q$, we start at the root node of the range tree on colors. Let $\nu$ denote the currently visited node.  If the query $Q$ to a data structure $\cE$ in a node $\nu$ returns {\tt yes}, then the number of distinct colors in $S(\nu)\cap Q$ does not exceed $\log^3 n$ and we answer the query using the data structure of Lemma~\ref{lemma:2dcapped-4sid2}. Otherwise we visit both children of $\nu$. The total number of visited nodes does not exceed $O(k/\log n)$. Hence the total time needed to answer a query is $O(\log n/\log\log n + k\log\log n)$. The total space usage of the data structure is increased by a factor $O(\log n)$ in comparison to \ref{lemma:2dcapped-4sid2}. This completes the proof of Theorem~\ref{theor:2dcapped-domin}.


\subsection{Type-2 Counting on a Small Set}
\label{sec:small}
It remains to show how to answer a type-2 range counting query in the case when set  contains a poly-logarithmic number of points of weight at most $\tau^2=\log^6 n$. 

Consider a set $S$ that contains at most $(1/18)\log n/\log \log n$ points. If coordinates and colors  are reduced to rank space, then there are at most $n^{1/2}$ combinatorially different sets of that size.  We can ask a poly-logarithmic number of different queries and the answer to each query has poly-logarithmic size.  Hence answers to all queries on all different sets can be stored in a universal look-up table of size $n^{1/2}\mathop{\rm polylog}n$.  

If $S$ contains $\log^{O(1)} n$ points, we use the recursive grid approach. Our method is the same as in Sections~\ref{sec:dominance} and~\ref{sec:4sided} with the following minor modifications: (1) The set of points is divided into $(1/4)\sqrt{\log n / (\tau'\log \log n)}$ columns and  $(1/5)\sqrt{\log n / (\tau'\log \log n)}$ rows for $\tau'=(\log \log n)^3$. 
(2) The top data structure is implemented using a look-up table as described above. 
(3) We do not use reduction to rank space on columns and rows.  Coordinates and colors of points are stored in a global rank space (i.e., the rank space of the set $S$). 
Recursion stops when the number of points in a set does not exceed $(1/18)\log n/\log \log n$. 

If a query range contains at most $\tau'$ colors, the query is answered as described in Section~\ref{sec:4sided}. For the case when the query range contains more than  $\tau'$ colors, we construct the range tree on colors as in Section~\ref{sec:general}. This increases the space usage by $O(\log\log n)$ factor. 
\begin{lemma}
  If a set $S$ in $\R^2$ contains $\log^c n$ points for a constant $c$ and has weight at most $\tau^2=\log^6 n$, then there exists a data structure that uses $O(|S|\log^2\log n)$ bits and answers  $4$-sided  type-2 counting queries in $O(1+ k(\log\log n/\log n))$ time.
\end{lemma}


We conclude:

}

\begin{theorem}
  \label{theor:2dcapped-domin}
  Let $S$ be the set of $m$ points in $\R^2$ with total weight $n\ge m$. There exists a data structure that uses $O(m \log m \log^{\eps} n)$ words of space and supports $4$-sided  type-2 range counting queries in $O(\log n/\log \log n + k\log\log n)$ time. 
\end{theorem}

{\small
\bibliographystyle{abbrv}
\bibliography{col_rs-with-counting}
}

\end{document}